\documentclass[a4paper, 12pt, reqno]{amsart}

\numberwithin{equation}{section}

\usepackage{amssymb}
\usepackage{mathrsfs}
\usepackage{dsfont}
\usepackage{stmaryrd, MnSymbol}
\usepackage{enumerate, xspace}
\usepackage{changebar}
\usepackage{hyperref}
\usepackage{pdfsync}

\newtheorem{theorem}{Theorem}
\newtheorem{lemma}{Lemma}[section]

\newtheorem{proposition}{Proposition}[section]

\newtheorem{remark}{Remark}[section]
\newtheorem{example}{Example}[section]

\numberwithin{equation}{section}

\renewcommand{\a}{\alpha}
\renewcommand{\b}{\beta}
\newcommand{\e}{\varepsilon}

\def\R{{\mathbb{R}}}
\def\N{{\mathbb{N}}}
\def\Z{{\mathbb{Z}}}

\begin{document}
\title[Energy Diffusion]
{thermal conductivity for stochastic energy exchange models}
\author{Makiko Sasada}
\address{Makiko Sasada\\
Graduate School of Mathematical Sciences, The University of Tokyo\\
3-8-1, Komaba, Meguro-ku, Tokyo, 153-8914, Japan}
\email{{\tt sasada@ms.u-tokyo.ac.jp}}

\begin{abstract}
We consider a class of stochastic models for energy transport and study relations between the thermal conductivity and some static observables, such as the static conductivity, which is defined as the contribution of static correlations in Green-Kubo formula. The class of models is a generalization of two specific models derived by Gaspard and Gilbert as mesoscopic dynamics of energies for two-dimensional and three-dimensional locally confined hard-discs. They claim some equalities hold between the thermal conductivity and several static observables and also conjecture that these equations are universal in the sense that they hold for mesoscopic dynamics of energies for confined particles interacting through hard-core collisions. In this paper, we give sufficient and necessary conditions for these equalities to hold in the class we introduce. In particular, we prove that the equality between the thermal conductivity and other static observables holds if and only if the model obeys the gradient condition. Since the gradient condition does not hold for models derived by Gaspard and Gilbert, our result implies a part of their claim is incorrect. 

\end{abstract}


\thanks{ This paper has been partially supported by by the Grant-in-Aid for Young Scientists (B) 25800068.}

\maketitle

\section{Introduction}\label{aba:sec1}
The derivation of the heat equation or the Fourier law for the macroscopic evolution of the energy through a space-time scaling limit from a microscopic dynamics given by Hamilton dynamics, is one of the most important problem in non-equilibrium statistical mechanics. Recently, Gaspard and Gilbert approached this problem by considering the motion of locally confined hard discs as microscopic dynamics. In a suitable limiting regime, they derive a stochastic model for energy transport as a mesoscopic dynamics and also study the thermal conductivity for this stochastic system. 

Let us recall the stochastic system obtained by Gaspard and Gilbert in \cite{GG08,GG09} (we call them GG-models). Consider a one-dimensional lattice and site energies $\e_1,\dots,\e_N$, $\e_i >0, \ 1 \le i \le N$. They evolve by a stochastic exchange of energy among pairs of neighboring sites. Let $P_N(\e_1,\dots,\e_N,t)$ denote the time-dependent energy distribution associated with the system. In the GG models, it is governed by the following master equation:
\begin{align}
& \partial_t  P_N(\e_1,\dots,\e_N,t)= \nonumber \\
   \sum_{i=1 }^{N-1} & \int_{-\e_{i+1}}^{\e_i} d\eta[W(\e_i-\eta,\e_{i+1}+\eta| \e_i,\e_{i+1})P_N(\dots,\e_i-\eta,\e_{i+1}+\eta, \dots,t)  \nonumber \\
 & - W(\e_i,\e_{i+1}| \e_i-\eta,\e_{i+1}+\eta)P_N(\dots,\e_i,\e_{i+1}, \dots,t) ],\label{eq:maindynamics}
\end{align}
where $W(\e_a,\e_b| \e_a-\eta,\e_b+\eta)$ describes the rate of exchange of energy $\eta$ between sites $a$ and $b$ at respective energies $\e_a$ and $\e_b$. In other words, in this dynamics, the amount of energy $\eta$ is moved between the neighboring sites $a$ and $b$ with rate $W(\e_a,\e_b| \e_a-\eta,\e_b+\eta)$. 

The microscopic mechanical dynamics consist of a one-dimensional array of two-dimensional cells, each containing a single hard-disc particle or an array of three-dimensional cells, each containing a single hard-sphere particle. The particles are confined to their cell, but also they can exchange energy with their neighbors. In these systems, the local dynamics in each cell and the interacting dynamics between neighbors are distinguished and energy transfer occurs only by the binary collisions. With a suitable limiting technique, the time-scale separation between the local dynamics and energy transport is achieved, which means each particle typically performs a large number of wall collisions with the wall of its confining cell, before it collides with one of its neighboring particles. In this regime, for the two-dimensional dynamics, Gaspard and Gilbert obtain the rate kernel 
\begin{align}\label{eq:2dimkernel}
W_{GG,2}(\e_a,\e_b| & \e_a-\eta,\e_b+\eta)  \nonumber \\
& =  \sqrt{\frac{2}{\pi^3}} \times 
\begin{cases}
\sqrt{\frac{1}{\e_a}}K(\sqrt{\frac{\e_b+\eta}{\e_a}}) \quad (-\e_b<\eta< \e_a-\e_b)  \\
\sqrt{\frac{1}{\e_b + \eta}}K(\sqrt{\frac{\e_a}{\e_b+\eta}}) \quad (\e_a-\e_b  <\eta< 0) \\
\sqrt{\frac{1}{\e_b}}K(\sqrt{\frac{\e_a -\eta}{\e_b}}) \quad (0 <\eta<\e_a)  \\
\end{cases}
\end{align}
valid for $\e_a \le \e_b$ where $K(m)$ is the Jacobi elliptic function of the first kind, i.e. $K(m)=\int_0^{\frac{\pi}{2}} \frac{1}{\sqrt{1-m^2 \sin ^2\theta}} d\theta$. In case $\e_a \ge \e_b$, by symmetry $W_{GG,2}$ is given as $W_{GG,2}(\e_a,\e_b| \e_a-\eta,\e_b+\eta)=W_{GG,2}(\e_b,\e_a| \e_b+\eta,\e_a-\eta)$.

Using the same strategy, Gaspard and Gilbert also derive the rate kernel for the three-dimensional dynamics as
\begin{align}\label{eq:3dimkernel}
W_{GG,3}(\e_a,\e_b| & \e_a-\eta,\e_b+\eta) \nonumber \\
& = \sqrt{\frac{\pi}{8}} \times
\begin{cases}
\sqrt{\frac{\e_b+\eta}{\e_a\e_b}} \quad (-\e_b<\eta<\min\{0,\e_a-\e_b\}) \\
\sqrt{\frac{1}{\max\{\e_a,\e_b\}}} \quad (\min\{0,\e_a-\e_b\} <\eta<\max\{0,\e_a-\e_b\}) \\
\sqrt{\frac{\e_a-\eta}{\e_a\e_b}} \quad (\max\{0,\e_a-\e_b\} <\eta<\e_a). \\
\end{cases}
\end{align}

The dynamics defined through (\ref{eq:maindynamics}) conserves the total energy, $\sum_{j=1}^N \e_j$. One also checks by inspection that stationary measures have the densities 
\begin{equation}\label{eq:equilibriummeasurea}
\prod_{i=1}^N \nu^T_{eq}(\e_i), \quad \text{for all} \quad T>0
\end{equation}
with 
\begin{equation}\label{eq:equilibriummeasure}
\nu^T_{eq}(\e)=\frac{\e^{\frac{d}{2}-1}\exp(-\frac{\e}{T})}{T^{\frac{d}{2}}\Gamma(\frac{d}{2})}.
\end{equation}
Here $d$ is the dimension of a single cell of the microscopic system, namely $d=2$ for the rate kernel $W_{GG,2}$ and $d=3$ for $W_{GG,3}$. The densities (\ref{eq:equilibriummeasurea}),(\ref{eq:equilibriummeasure}) are inherited from the underlying mechanical system: their invariant measure is Lebesgue for the positions and independent Maxwellian at temperature $T$ for the velocities. 

Note that we set unity such that the Boltzmann constant $k_B \equiv 1$ and the mass of each particle is also one. We denote the equilibrium average at temperature $T$ by $\langle \ \cdot \ \rangle_T$:
\begin{align}
 \langle f (\e_1, \dots, \e_N) \rangle_T  =\int_0^{\infty} \dots \int_0^{\infty}  f (\e_1, \dots, \e_N)   \prod_{i=1}^N \nu_{eq}(\e_i) d \e_i.
\end{align}

The thermal conductivity $\kappa(T)$ associated with energy transport is characterized by Fourier's law
\begin{equation}
q= -\kappa(T)\partial_xT,
\end{equation}
where $q$ is the heat flux and $\partial_x T$ is the temperature gradient. More precisely, under the diffusive space-time scaling limit, the time evolution of the local temperature
\begin{equation}
T \big(\frac{a}{N},t\big)=\frac{2}{d}\int_0^{\infty} \dots \int_0^{\infty}  \e_aP(\e_1,\dots,\e_N,tN^2)  \ \prod_{i=1}^N d\e_i 
\end{equation}
will be given by 
\begin{equation}\label{eq:Fourier}
\partial_t T=\partial_x (D(T)\partial_x T), \quad  T=T(x,t)
\end{equation}
with thermal diffusivity $D(T)=\frac{T^2}{\chi_T}\kappa(T)$. $\chi_T$ is the static energy variance $\chi_T=\langle \e_0^2 \rangle_T - \langle \e_0 \rangle_T^2$. $\chi_T=\frac{1}{2}dT^2$ by the ideal gas law $\langle \e \rangle_T=\frac{1}{2}dT$. One of the main interest is the relation between $\kappa(T)$ and the rate kernels $W_{GG,2}$ and $W_{GG,3}$.

To reveal the relation, Gaspard and Gilbert consider following static observables: $\langle \nu(\e_a,\e_b) \rangle_T$, $\langle (\e_a-\e_b)j(\e_a,\e_b) \rangle_T$ and $\langle h(\e_a,\e_b) \rangle_T$  where $\nu$, $j$ $h$ are the collision frequency, the first and second moments of the energy exchanges respectively defined as
\begin{align}
\nu(\e_a,\e_b) & :=\int_{-\e_b}^{\e_a} W_{GG,d}(\e_a,\e_b|\e_a-\eta,\e_b+\eta) d\eta, \\
j(\e_a,\e_b) & :=\int_{-\e_b}^{\e_a}\eta W_{GG,d}(\e_a,\e_b|\e_a-\eta,\e_b+\eta) d\eta, \\
h(\e_a,\e_b) & :=\int_{-\e_b}^{\e_a}\eta^2 W_{GG,d}(\e_a,\e_b|\e_a-\eta,\e_b+\eta) d\eta
\end{align}
for $d=2,3$. By symmetry, $\nu(\e_a,\e_b)=\nu(\e_b,\e_a)$, $j(\e_a,\e_b)=-j(\e_b,\e_a)$ and $h(\e_a,\e_b)=h(\e_b,\e_a)$. Note that $j(\e_a,\e_b)$ is the current of the energy from
the site $a$ to $b$. They also define $\kappa_s(T)$ as the conductivity obtained by the contributions of static correlations in Green-Kubo formula and derive the following equations in \cite{GG08, GG09} (cf. the equations (8)-(10) in \cite{GG09}):
\begin{equation}\label{eq:staticrelations}
\kappa(T)=\kappa_s(T)=\frac{1}{2} \langle (\e_a-\e_b)j(\e_a,\e_b)\rangle_T = \frac{1}{2} \langle h(\e_a,\e_b)\rangle_T =\langle \nu(\e_a,\e_b)\rangle_T.  
\end{equation}
Moreover, they conjecture that (\ref{eq:staticrelations}) holds for more general stochastic energy exchange models obtained from locally confined collision dynamics (cf. Section 8 in \cite{GG09}). However, by the simple scaling argument, we have the following scaling relations
\begin{align}\label{eq:scales1}
& \kappa(T)  = \kappa \sqrt{T}, \quad \kappa_s(T)=\kappa_s\sqrt{T}, \quad \frac{1}{2}<(\e_a-\e_b) j(\e_a,\e_b) >_T  =\kappa_{1}T^{5/2},  \nonumber \\
&  \frac{1}{2}<h(\e_a,\e_b)>_T = \kappa_{2}T^{5/2} \quad <\nu(\e_a,\e_b)>_T = \kappa_{f} \sqrt{T} 
\end{align}
for constants $\kappa,\kappa_s,\kappa_1,\kappa_2$ and $\kappa_f$ which are independent from $T$. Namely, the precise conjecture should be 
\begin{equation}\label{eq:realconjecture}
\kappa=\kappa_s=\kappa_1=\kappa_2=\kappa_f
\end{equation}
holds where
\begin{align}
& \kappa  =\kappa(1), \quad \kappa_s=\kappa_s(1), \quad \kappa_1=\frac{1}{2}<(\e_a-\e_b) j(\e_a,\e_b) >_1,  \nonumber \\
&  \kappa_2=\frac{1}{2}<h(\e_a,\e_b)>_1, \quad \kappa_f=<\nu(\e_a,\e_b)>_1.
\end{align}

The aim of the paper is to find the sufficient and necessary condition for (\ref{eq:realconjecture}) to hold. We first introduce a generalized class of rate kernels, which we call mesoscopic mechanical kernels including $W_{GG,2}$ and $W_{GG,3}$. Then, we show $\kappa_s=\kappa_1=\kappa_2$ holds universally in this class. We also show that $\kappa_1=\kappa_f$ holds under the additional condition (\ref{eq:3=4}). It is easy to check (\ref{eq:3=4}) holds for $W_{GG,2}$ and $W_{GG,3}$. Finally, we show $\kappa=\kappa_s$ holds if the current $j(\e_a,\e_b)$ is written as a gradient form and otherwise $\kappa < \kappa_s$. Since the current for the rate kernels $W_{GG,2}$ and $W_{GG,3}$ do not have a gradient form, we conclude $\kappa < \kappa_s=\kappa_1=\kappa_2=\kappa_f$ for GG models. We remark that the simulation result in \cite{GG08,GG09} may imply the difference between $\kappa$ and $\kappa_s$ is very small compared to $\kappa$ (or $\kappa_s$).

\section{Mesoscopic mechanical kernels}

We recall the model introduced in the last section. We consider a system of $N$ energy cells $\e_1,\dots,\e_N$ along a one-dimensional axis with stochastic exchange of energy among pairs of neighboring cells. The dynamics is given by the master equation for the time-dependent energy distribution $P_N(\e_1,\dots,\e_N,t)$ in terms of the kernel $W(\e_a,\e_b| \e_a-\eta,\e_b+\eta)$.

We call $W$ is a mesoscopic mechanical kernel if $W$ satisfies the following three conditions:
\begin{enumerate}
\item[(i)] \  $W(c  \e_a, c  \e_b|  c(\e_a- \eta), c (\e_b+\eta)) = \sqrt{\frac{1}{c}} W(\e_a,\e_b|\e_a-\eta,\e_b+\eta), \quad \forall c>0$, 
 \item[(ii)] \ $W(\e_a,\e_b|\e_a-\eta,\e_b+\eta) =W(\e_b,\e_a|\e_b+\eta,\e_a-\eta)$, 
  \item[(iii)] \ $(\e_a\e_b)^{\frac{d}{2}-1}  W(\e_a,\e_b|\e_a-\eta,\e_b+\eta) \\
 = ((\e_a-\eta)(\e_b+\eta))^{\frac{d}{2}-1} W(\e_a-\eta,\e_b+\eta|\e_a,\e_b)$ for some $d \in \N$.
\end{enumerate}
If the rate kernel $W$ is obtained as a mesoscopic energy transport model of mechanical models in which locally confined particles interact through hard-core collisions, the conditions (i) to (iii)  are automatically satisfied. 
Property (i) comes from the observation that the energy exchange happens with rate proportional to
the absolute value of the velocity of microscopic hard particles, which is the square root of the energy. 
The property (ii)
claims a dynamical symmetry between the sites $a$ and $b$. The property (iii) is the detailed balance condition which follows from the reversibility of the underling mechanical dynamics. Under the condition (iii), the equilibrium measure for the stochastic energy exchange model at temperature $T$ has density $\prod_{i=1}^N \nu_{eq}^T(\e_i)$ with $\nu_{eq}^T$ defined by (\ref{eq:equilibriummeasure}). 

We call the stochastic energy exchange model given by the master equation (\ref{eq:maindynamics}) is a mesoscopic mechanical model if its kernel is a mesoscopic mechanical kernel. In the rest of the paper, we only consider the mescoscopic mechanical models.

\begin{remark}
In \cite{GKS}, more general class of energy exchange models are introduced. Also, they introduce its subclass, called mechanical models which are closely relates to the mesocopic mechanical models which we introduced here.
\end{remark}

\section{Green-Kubo formula}

In this section, we given explicit expressions of $\kappa(T)$ and $\kappa_s(T)$ by Green-Kubo formula. The content of this section is quite classical and general (see \cite{Sp}).  

From (\ref{eq:Fourier}), $\kappa(T)$ equals the spreading of energy from an initial perturbation, modulo the proper normalization. This leads 
\begin{align}
\kappa(T) & =  - \frac{1}{4T^2} \lim_{t \to \infty} \frac{1}{t} \sum_{i \in \Z} i^2 \langle (\e_i(t)-\e_i(0))(\e_0(t)- \e_0(0)) \rangle_{\mathbb{P}_T}, 
\end{align}
since by the symmetry of the dynamics and the form of the equilibrium densities
\begin{align}
& \langle (\e_i(t)-\e_i(0))(\e_0(t)- \e_0(0)) \rangle_{\mathbb{P}_T} = \nonumber \\
& = \langle \e_i(t)\e_0(0) \rangle_{\mathbb{P}_T} +  \langle \e_i(0)\e_0(t) \rangle_{\mathbb{P}_T}  - \langle \e_i(t) \rangle_{\mathbb{P}_T} \langle \e_0(t) \rangle_{\mathbb{P}_T} - \langle \e_i(0) \rangle_{\mathbb{P}_T} \langle \e_0(0) \rangle_{\mathbb{P}_T}\nonumber \\
& = 2 \big( \langle \e_i(t)\e_0(0) \rangle_{\mathbb{P}_T} - \langle \e_i(0) \rangle_{\mathbb{P}_T} \langle \e_0(0) \rangle_{\mathbb{P}_T} \big) =2 \langle (\e_i(t)-\e_i(0))\e_0(0) \rangle_{\mathbb{P}_T} 
\end{align}
where $\mathbb{P}_T$ is the distribution of the dynamics starting from the equilibrium state at temperature $T$. 


Now, we rewrite $\e_i(t)-\e_i(0)=M_i(t)+J_{i-1,i}(t)-J_{i,i+1}(t)$ where $M_i(t)$ is a martingale and $J_{i,i+1}(t)=\int_0^t j(\e_i(s),\e_{i+1}(s)) ds$. Then, we have
\begin{align}
 &\sum_{i \in \Z} i^2 \langle (\e_i(t)-\e_i(0))(\e_0(t)- \e_0(0)) \rangle_{\mathbb{P}_T} \nonumber  \\
 & =  \sum_{i \in \Z} i^2 \langle (M_i(t)+J_{i-1,i}(t)-J_{i,i+1}(t))(M_0(t)+J_{-1,0}(t)-J_{0,1}(t)) \rangle_{\mathbb{P}_T} \nonumber \\
 &=  \sum_{i \in \Z} i^2 \langle M_i(t) M_0(t) \rangle_{\mathbb{P}_T}  +  \sum_{i \in \Z} i^2 \langle (J_{i-1,i}(t)-J_{i,i+1}(t) ) (J_{-1,0}(t)-J_{0,1}(t)) \rangle_{\mathbb{P}_T}
\end{align}
since $\langle M_i(t) J_{j-1,j}(t)\rangle_{\mathbb{P}_T}=0$ for any $i,j$. Now, since $M_i(t)$ has independent increments, we have
\begin{align}
 &  \langle M_i(t) M_0(t) \rangle_{\mathbb{P}_T} \nonumber \\
  & = -t \langle \e_i \big(j(\e_{-1},\e_0)-j(\e_0,\e_1) \big) + \big(j(\e_{i-1},\e_i)-j(\e_i,\e_{i+1}) \big)  \e_0 \rangle_{T}   \nonumber \\
 & = -2t \langle \e_i \big(j(\e_{-1},\e_0)-j(\e_0,\e_1) \big) \rangle_T. 
\end{align}
Then, by the summation by parts, 
\begin{align}
\sum_{i \in \Z} i^2  \langle M_i(t) M_0(t) \rangle_{\mathbb{P}_T}  = 2t \sum_{i \in \Z} (2i-1) \langle \e_i j(\e_0,\e_1) \rangle_1=4t \sum_{i \in \Z}i \langle \e_i j(\e_0,\e_1) \rangle_T
\end{align}
since $\langle \e_i j(\e_0,\e_1) \rangle_T = -\langle \e_{-i+1} j(\e_0,\e_1) \rangle_T$ for all $i$. In the same manner,
\begin{align}
\sum_{i \in \Z} i^2 \langle (J_{i-1,i}(t)-J_{i,i+1}(t) ) (J_{-1,0}(t)-J_{0,1}(t)) \rangle_{\mathbb{P}_T} = -2 \sum_{i \in \Z}  \langle J_{i,i+1}(t) J_{0,1}(t) \rangle_{\mathbb{P}_T}.
\end{align}
Therefore, the Green-Kubo formula can be rewritten in the sum of the \lq\lq static" term and the \lq\lq dynamic" term as
\begin{align}
& \kappa(T)   \nonumber \\
& = - \frac{1}{T^2} \sum_{i \in \Z} i \langle \e_i j(\e_0,\e_1) \rangle_1 - \frac{1}{T^2}\int_0^{\infty}   \sum_{i \in \Z} \langle j (\e_ {i} (0),\e_{i+1} (0) )  e^{tL} j (\e_0 (0),\e_1 (0)) \rangle_{\mathbb{P}_T}   dt \nonumber\\
 & = \kappa_s(T) - \frac{1}{T^2} \int_0^{\infty}   \sum_{i \in \Z} \langle j (\e_ {i} (0),\e_{i+1} (0) )  e^{tL} j (\e_0 (0),\e_1 (0)) 
   \rangle_{\mathbb{P}_T}   dt
\end{align}
where $L$ is the formal infinite volume generator for the stochastic process given by rate $W$.  
Here, we denote the static part by $\kappa_s(T)$ which is defined as 
\begin{align}\label{eq:static}
\kappa_s(T)= - \frac{1}{T^2}\sum_{i \in \Z} i \langle \e_i j(\e_0,\e_1) \rangle_T.
\end{align}
Since $L$ is a non-positive operator, the \lq\lq dynamic" term is non positive, hence $\kappa(T) \le \kappa_s(T)$ is generally true. 

\section{Main results}
In this section, we give our main result for the relations between the thermal conductivity and the static observables appearing in the conjecture (\ref{eq:staticrelations}). 
\begin{lemma}
For any mesoscopic mechanical model, 
\begin{align}\label{eq:scales}
& \kappa(T)  = \kappa \sqrt{T}, \quad \kappa_s(T)=\kappa_s\sqrt{T}, \quad \frac{1}{2}<(\e_a-\e_b) j(\e_a,\e_b) >_T  =\kappa_{1}T^{5/2},  \nonumber \\
&  \frac{1}{2}<h(\e_a,\e_b)>_T = \kappa_{2}, T^{5/2} \quad <\nu(\e_a,\e_b)>_T = \kappa_{f} \sqrt{T}
\end{align}
where
\begin{align}
& \kappa  =\kappa(1), \quad \kappa_s=\kappa_s(1), \quad \kappa_1=\frac{1}{2}<(\e_a-\e_b) j(\e_a,\e_b) >_1,  \nonumber \\
&  \kappa_2=\frac{1}{2}<h(\e_a,\e_b)>_1, \quad \kappa_f=<\nu(\e_a,\e_b)>_1.
\end{align}
\end{lemma}
\begin{proof}
Note that he law of $(T\e_i)_{i=1}^N$ under $\nu^1_{eq}$ is equal to the law of $(\e_i)_{i=1}^N$ under $\nu^T_{eq}$. Then, by the condition (i), we have the scaling relations except for $\kappa(T)$. For $\kappa(T)$, by the scaling relation (i) again, we see that the distribution of $\e(t)=(\e_i(t))_{i \in \Z}$ under $\mathbb{P}_T$ is same as that of $T \e(\sqrt{T}t)=(T\e_i(\sqrt{T}t))_{i \in \Z}$ under $\mathbb{P}_1$. Therefore, we have
\begin{align}
\kappa(T)= & \lim_{t \to \infty} \frac{1}{t} \sum_{i \in \Z} i^2 \langle (\e_i(t)-\e_i(0))(\e_0(t)- \e_0(0)) \rangle_{\mathbb{P}_T} \nonumber \\
& =T^2 \lim_{t \to \infty}  \frac{1}{t}\sum_{i \in \Z} i^2 \langle (\e_i(\sqrt{T}t)-\e_i(0))(\e_0(\sqrt{T}t)- \e_0(0)) \rangle_{\mathbb{P}_1}  \nonumber \\
& = T^{5/2} \lim_{t \to \infty} \frac{1}{t} \sum_{i \in \Z} i^2 \langle (\e_i(t)-\e_i(0))(\e_0(t)- \e_0(0)) \rangle_{\mathbb{P}_1}. 
\end{align}
The second equality follows from the change of variable from $\sqrt{T}t$ to $t$. 
\end{proof}

\begin{proposition}\label{prop:1}
For any mesoscopic mechanical model,
\begin{equation}
\kappa_s=\kappa_1=\kappa_2.
\end{equation}
\end{proposition}
\begin{proof}
The first equality follows from the simple observation 
\begin{align}\label{eq:static}
\kappa_s= - \sum_{i \in \Z} i \langle \e_i j(\e_0,\e_1) \rangle_1
= -  \langle \e_1 j(\e_0,\e_1) \rangle_1 =\frac{1}{2}\langle (\e_0-\e_1) j(\e_0,\e_1) \rangle_1
\end{align}
since the equilibrium measure is product and the current is antisymmetric. The second equality is shown in Appendix B.
\end{proof}

\begin{proposition}\label{prop:2}
Assume $W$ is a mesoscopic mechanical kernel. Then $\kappa_1=\kappa_f$ holds if and only if the kernel $W$ satisfies
\begin{align}\label{eq:3=4}
\int_0^1 \int_0^1 &  \tilde{W}(\a, \b) d\b d\a   =(d+\frac{3}{2})(d+\frac{1}{2})\int_0^1 \int_0^1 \a (\a-\beta)  \tilde{W}(\a, \b) d\b d\a
\end{align}
where $\tilde{W}(\a, \b)=W(\a, 1-\a|\beta,1-\beta ) (\a(1-\a))^{\frac{d}{2} -1}$. 
\end{proposition}
We give its proof in Appendix B again. Note that by the assumptions (ii) and (iii), we have $\tilde{W}(\a,\b)=\tilde{W}(\b,\a)=\tilde{W}(1-\a,1-\b)=\tilde{W}(1-\b,1-\a)$. We can check $W_{GG,2}$ and $W_{GG,3}$ satisfy the condition (\ref{eq:3=4}) by the direct computation as follows.

\begin{example}[GG-model, $d=2$, \cite{GG08}]
By the definition of the kernel,
\begin{align}
& \sqrt{\frac{\pi^3}{2}} \tilde{W}_{GG,2}(\a, \b) \nonumber \\
&=  \sqrt{\frac{1}{ ( \a \wedge \b) \vee (  (1-\a) \wedge (1-\b) )}} 
K \Big(\sqrt{\frac{\a \wedge (1-\a) \wedge \b \wedge (1-\b) }{( \a \wedge \b) \vee (  (1-\a) \wedge (1-\b) )}}\Big). 
\end{align}
Then, (\ref{eq:3=4}) is equivalent to 
\begin{align}
& \int_0^1 \int_0^1 \int_0^{\frac{\pi}{2}}  \frac{d\theta d\b d\a }{ \sqrt{ ( \a \wedge \b) \vee (  (1-\a) \wedge (1-\b) ) - \a \wedge (1-\a) \wedge \b \wedge (1-\b) \sin \theta^2}} \nonumber  \\
&  =\frac{35}{4}\int_0^1 \int_0^1 \int_0^{\frac{\pi}{2}}  \frac{\a (\a-\b) \ d\theta d\b d\a }{ \sqrt{ ( \a \wedge \b) \vee (  (1-\a) \wedge (1-\b) ) - \a \wedge (1-\a) \wedge \b \wedge (1-\b) \sin \theta^2}}
\end{align}
which we can show by the direct calculation. Therefore, by Proposition \ref{prop:1} and \ref{prop:2}, $\kappa_s=\kappa_1=\kappa_2=\kappa_f$ holds for the kernel $W_{GG,2}$.
\end{example}

\begin{example}[GG-model, $d=3$, \cite{GG09}]
In the same manner for the $d=2$ case, we have
\begin{align}
\sqrt{\frac{8}{\pi}} \tilde{W}_{GG,3}(\a,\b) = \sqrt{ \a \wedge (1-\a) \wedge \b \wedge (1-\b)}.
\end{align}
Then, (\ref{eq:3=4}) is equivalent to 
\begin{align}
& \int_0^1 \int_0^1  \sqrt{ \a \wedge (1-\a) \wedge \b \wedge (1-\b)} d\a d\b \nonumber  \\ 
&  =\frac{63}{4}\int_0^1 \int_0^1  \a (\a-\b) \sqrt{ \a \wedge (1-\a) \wedge \b \wedge (1-\b)} d\a d\b
\end{align}
which we can show by the direct calculation again. Therefore, by Proposition \ref{prop:1} and \ref{prop:2}, $\kappa_s=\kappa_1=\kappa_2=\kappa_f$ holds for the kernel $W_{GG,3}$.
\end{example}

\begin{remark}
We could not find any physical interpretation of the condition (\ref{eq:3=4}) so far.
\end{remark}

The next theorem shows the relation between the thermal conductivity and the other static observables.
\begin{theorem}\label{thm:1}
For any mesoscopic mechanical model,
$\kappa=\kappa_s$ if and only if 
\begin{align}\label{eq:gradientcurrent}
j(\e_a,\e_b)=C(\e_a^{3/2}-\e_b^{3/2})
\end{align}
for some constant $C$.
\end{theorem}
We give its proof in Appendix A.
Since the current of the GG model with the kernel (\ref{eq:2dimkernel}) or (\ref{eq:3dimkernel}) is not a gradient of any function, $\kappa < \kappa_s$. 
 


\section{Conclusions}

To summarize, we have introduced a class of rate kernels, which is a generalization of GG-models, and obtained a necessary and sufficient conditions for the equalities (\ref{eq:realconjecture}) to hold in this class. 

We showed that the equalities $\kappa_s=\kappa_1=\kappa_2$ hold universally in this class, though the equality between the static conductivity and the frequency of energy exchange holds only under the additional condition (\ref{eq:3=4}). 

We also showed that the equality between the thermal conductivity and the static conductivity holds if and only if the model satisfies the gradient condition. Also we see that if the gradient condition is satisfied, the current should be written an explicit form.

From the above observations, we conclude that $\sqrt{T} = \langle \nu(\e_a,\e_b) \rangle_T =   \kappa_s(T) >  \kappa (T)$ for the GG-models $W_{GG,2}$ and $W_{GG,3}$. 

\appendix

\section{Gradient condition and the thermal conductivity}
We give a proof for Theorem \ref{thm:1}. For general reversible stochastic models including our dynamics, the variational formula for the thermal conductivity is known as
\begin{align}
 \kappa= \frac{1}{2} \inf_{f} \langle \int^{\e_0}_{-\e_1} d\eta W(\e_0,\e_1| \e_0-\eta,\e_1+\eta) \Big(\eta + \nabla_{0,1,\eta}\Sigma_f \Big)^2 \rangle_1
\end{align}
where $\nabla_{0,1,\eta}\Sigma_f=\Sigma_f ( \dots, \e_0-\eta,\e_1+\eta, \dots) -\Sigma_f ( \dots, \e_0,\e_1, \dots)$, $\Sigma_f=\sum_{i \in \Z} \tau_i f$ and the infimum is taken over all cylinder functions $f$ (cf. \cite{DFGW}, \cite{Sp}). $\kappa_s$ is given by the choice $f=0$ in the above variational formula : 
\begin{align}\label{eq:static2}
\kappa_s = \frac{1}{2} < \int^{\e_0}_{-\e_1} d\eta W(\e_0,\e_1| \e_0-\eta,\e_1+\eta) \eta^2 >_1 =  \frac{1}{2} < h(\e_0,\e_1)>_1.
\end{align}

Then, we have
\begin{align}
2 & \kappa=  2 \kappa_s+ \inf_{f} \Big\{ - 2  \langle  \Sigma_{f} \ j(\e_0,\e_1)\rangle_1  \nonumber \\
 & +   \langle  \int^{\e_0}_{-\e_1} d\eta W(\e_0,\e_1| \e_0-\eta,\e_1+\eta) \Big(\nabla_{0,1,\eta}\Sigma_f  \Big)^2 \rangle_1 \Big\}
 \end{align}
so $\kappa=\kappa_s$ if and only if $\langle  \Sigma_{f} \ j(\e_0,\e_1)\rangle_1=0$ for any cylinder function $f$. Now, we give a key lemma.

\begin{lemma}
$\langle  \Sigma_{j(\e_0,\e_1)} \ j(\e_0,\e_1)\rangle_1= \langle   \big( j(\e_0,\e_1)+\tilde{j}(\e_1) -\tilde{j}(\e_0) \big)^2 \rangle_1$ where $\tilde{j}(\e_0):=\frac{1}{\Gamma(\frac{d}{2})}\int_0^{\infty} j(\e_0,x) x^{\frac{d}{2}-1}\exp(-x) dx$. In particular, $\langle  \Sigma_{j(\e_0,\e_1)} \ j(\e_0,\e_1)\rangle_1=0$ if and only if $j(\e_0,\e_1)= \tilde{j}(\e_0) -\tilde{j}(\e_1)$.
\end{lemma}
\begin{proof}
Since $j$ is antisymmetric and the equilibrium measure is translation invariant,
\begin{align}
  \langle & j(\e_{-1},\e_0) \ j(\e_0,\e_1) \rangle_1  =   \langle  j(\e_{1},\e_2)  \ j(\e_0,\e_1) \rangle_1 \nonumber \\
  & = \frac{1}{\Gamma(\frac{d}{2})^3}\int_{0}^{\infty}\int_{0}^{\infty} \int_{0}^{\infty} j(x,y)j(y,z)(xyz)^{\frac{d}{2}-1}\exp(-x-y-z)dxdydz \nonumber \\
   & =\frac{1}{\Gamma(\frac{d}{2})^2} \int_{0}^{\infty}\int_{0}^{\infty}  j(x,y)\tilde{j}(y) (xy)^{\frac{d}{2}-1}\exp(-x-y)dxdy  =  \langle  j(\e_0,\e_1) \tilde{j}(\e_1) \rangle_1 \nonumber \\
  & = - \frac{1}{\Gamma(\frac{d}{2})^2} \int_{0}^{\infty}\int_{0}^{\infty}  j(y,x)\tilde{j}(y) (xy)^{\frac{d}{2}-1}\exp(-x-y)dxdy = -  \langle  j(\e_0,\e_1) \tilde{j}(\e_0) \rangle_1 \nonumber \\
  & = - \frac{1}{\Gamma(\frac{d}{2})} \int_{0}^{\infty}  \tilde{j}(y)^2 y^{\frac{d}{2}-1}\exp(-y)dy =- \langle   \tilde{j}(\e_0)^2 \rangle_1
\end{align}
and $\langle   \tilde{j}(\e_0) \rangle_1=0$. Therefore,  we have
\begin{align}
 \langle  \Sigma_{j(\e_0,\e_1)} \ j(\e_0,\e_1) \rangle_1 & = \langle  (j(\e_{-1},\e_0) + j(\e_0,\e_1)+j(\e_1,\e_2) ) \ j(\e_0,\e_1) \rangle_1 \nonumber\\
 &= \langle  j(\e_0,\e_1) \ j(\e_0,\e_1)  \rangle_1 - 2\langle  \tilde{j}(\e_0)  \ \tilde{j}(\e_0) \rangle_1 \nonumber \\
& = \langle   \big( j(\e_0,\e_1)+\tilde{j}(\e_1) -\tilde{j}(\e_0) \big)^2 \rangle_1. 
\end{align} 
\end{proof}


\begin{proposition}\label{prop:A1}
$\kappa=\kappa_s$ if and only if the current is written as a gradient form.
\end{proposition}
\begin{proof}
By the above lemma, if $\kappa=\kappa_s$, then $ \langle  \Sigma_{j(\e_0,\e_1)},j(\e_0,\e_1) \rangle_1  = 0$ holds. Therefore  $ j(\e_0,\e_1)= \tilde{j}(\e_0) -\tilde{j}(\e_1)$. Namely, the current is a gradient of $\tilde{j}$. On the other hand, if the current is a gradient of some function $h$ then obviously, $\langle  \Sigma_{f},j(\e_0,\e_1)\rangle_1=0$ for any cylinder function $f$ and so $\kappa = \kappa_s$.
\end{proof}

\begin{remark}
By the above lemma, if $j(\e_0,\e_1)=h(\e_0)-h(\e_1)$ for some function $h$, then $j(\e_0,\e_1)=\tilde{j}(\e_0)-\tilde{j}(\e_1)$.
\end{remark}

Now, we give an explicit form of the current $j$ under the condition that the current is written as a gradient form.
\begin{lemma}
If $j(\e_0,\e_1)=h(\e_0)-h(\e_1)$ for some function $h$, then $j(\e_0,\e_1)=C(\e_0^{3/2}-\e_1^{3/2})$ for some constant $C$.
\end{lemma}
\begin{proof}
By the condition (i), $j(\e_0,\e_1)=(\e_0+\e_1)^{3/2}j(\frac{\e_0}{\e_0+\e_1},\frac{\e_1}{\e_0+\e_1})$. Then, since $j(\e_0,\e_1)=h(\e_0)-h(\e_1)$, for any $c>0$,
\begin{align}
c^{3/2} \big(h(\e_0)-h(\e_1)\big)=c^{3/2}j(\e_0,\e_1)=j(c\e_0,c\e_1)=h(c\e_0)-h(c\e_1).
\end{align}
As $h$ is defined up to constant, we can assume $h(0)=0$ without loss of generality, so $c^{3/2}h(\e_0)=h(c\e_0)$.
\end{proof}

\begin{remark}
Feng et al. \cite{FIS} studied the stochastic model given by (\ref{eq:maindynamics}) with the rate kernel $W(\e_a,\e_b|\e_a-\eta,\e_b+\eta)=\sqrt{\frac{1}{|\eta|}}$. It is easy to see that the kernel satisfies the conditions (i),(ii) and (iii) with $d=2$ and its current is written as a gradient form. 
\end{remark}

\begin{remark}
Proposition \ref{prop:A1} claims that the variational formula for the thermal conductivity attains its infimum with $f \equiv 0$ if and only if the current can be written as a gradient form. This property should hold for very general models (cf. Theorem 1.1 in \cite{Li}). In fact, in our proof, we use only the following three properties of the model: the current depends only on $\e_0$ and $\e_1$, the current is antisymmetric and the equilibrium measure is product.
\end{remark}

\section{static values}

We give a proof of Proposition \ref{prop:1} and \ref{prop:2}. For this purpose, we give some explicit expressions for static values, $\kappa_f,\kappa_1$ and $\kappa_2$.
First, we introduce a function $\bar{W}(\a, \b) : (0,1)^2 \to \R_+$ defined as $\bar{W}(\a,\b):=W(\a,1-\a|\b,1-\b)$. Then, the assumptions (i), (ii) and (iii) are rewritten in the following way :
\begin{enumerate}
\item[(i)] $W(\e_a,\e_b|\e_a-\eta,\e_b+\eta) = \sqrt{\frac{1}{\e_a+\e_b}} \bar{W}(\frac{\e_a}{\e_a+\e_b},\frac{\e_a -\eta}{\e_a+\e_b})$, 
\item[(ii)] $\bar{W}(\a,\b)=\bar{W}(1-\a,1-\b)$,
\item[(iii)] $(\a(1-\a))^{\frac{d}{2}-1}\bar{W}(\a,\b) = (\b(1-\b))^{\frac{d}{2}-1} \bar{W}(\b,\a)$.
\end{enumerate}

Using the notation $\bar{W}$, we have 
\begin{align}
\nu(\e_a,\e_b) & = \sqrt{\e_a+\e_b} \int_{0}^{1} \bar{W}(\frac{\e_a}{\e_a+\e_b}, \beta )d\beta, \\
j(\e_a,\e_b) & =(\sqrt{\e_a+\e_b})^3 \int_{0}^{1} \big(\frac{\e_a}{\e_a+\e_b} -\beta \big) \bar{W} \big(\frac{\e_a}{\e_a+\e_b}, \beta \big)d\beta, \\
h(\e_a,\e_b) & =(\sqrt{\e_a+\e_b})^5 \int_{0}^{1} \big(\frac{\e_a}{\e_a+\e_b} -\beta \big)^2 \bar{W}\big(\frac{\e_a}{\e_a+\e_b}, \beta \big)d\beta.
\end{align}
Then, since $\e_a+\e_b$ and $\frac{\e_a}{\e_a+\e_b}$ are independent under the equilibrium measure, which is the product of identical gamma distributions, we have
\begin{align}
\langle \nu(\e_a,\e_b) \rangle_1 & = \langle\sqrt{\e_a+\e_b} \rangle_1 \langle \int_{0}^{1} \bar{W}(\frac{\e_a}{\e_a+\e_b}, \beta )d\beta  \rangle_1,  \nonumber \\
\langle (\e_a-\e_b) j(\e_a,\e_b) & \rangle_1  \nonumber \\
=\langle(\sqrt{\e_a+\e_b})^5  \rangle_1 &  \langle \int_{0}^{1} \big(\frac{2\e_a}{\e_a+\e_b} -1 \big) \big(\frac{\e_a}{\e_a+\e_b} -\beta \big) \bar{W} \big(\frac{\e_a}{\e_a+\e_b}, \beta \big)d\beta  \rangle_1, \nonumber \\
\langle h(\e_a,\e_b)  \rangle_1 & =\langle(\sqrt{\e_a+\e_b})^5 \rangle_1 \langle \int_{0}^{1} \big(\frac{\e_a}{\e_a+\e_b} -\beta \big)^2 \bar{W}\big(\frac{\e_a}{\e_a+\e_b}, \beta \big)d\beta  \rangle_1.
\end{align}
By the definition, we have $\langle (\e_a+\e_b)^m \rangle_1=\frac{\Gamma(d+m)}{\Gamma(d)}$, and for any measurable $f:(0,1) \to \R$, 
\begin{align}
\langle \int_{0}^{1} &  f(\frac{\e_a}{\e_a+\e_b})  \bar{W}(\frac{\e_a}{\e_a+\e_b}, \beta )d\beta  \rangle_1 
\nonumber \\&  =\frac{\Gamma(d)}{\Gamma(\frac{d}{2})^2}\int_{0}^{1}\int_{0}^{1}   f(\a)  \bar{W}(\a, \beta )d\beta  (\a(1-\a))^{\frac{d}{2}-1} d\a \nonumber \\
&  =\frac{\Gamma(d)}{\Gamma(\frac{d}{2})^2}\int_{0}^{1}\int_{0}^{1}   f(\a)  \tilde{W}(\a, \beta )d\beta d\a
 \end{align}
where $\tilde{W}(\a, \beta )=\bar{W}(\a, \beta )(\a(1-\a))^{\frac{d}{2}-1}$.
Therefore, we have 
where
\begin{align}\label{eq:constants}
\kappa_f & =\frac{\Gamma(d+\frac{1}{2})}{\Gamma(\frac{d}{2})^2}\int_0^1 \int_0^1 \tilde{W}(\a, \beta )  d\beta d\a, \nonumber \\
\kappa_1 & =\frac{\Gamma(d+\frac{5}{2})}{\Gamma(\frac{d}{2})^2}\int_0^1 \int_0^1 (\a-\frac{1}{2})(\a-\beta) \tilde{W}(\a, \beta ) d\beta d\a, \nonumber \\
\kappa_2 & =\frac{\Gamma(d+\frac{5}{2})}{ \Gamma(\frac{d}{2})^2}\int_0^1 \int_0^1 \frac{(\a-\beta)^2}{2} \tilde{W}(\a, \beta ) d\beta d\a.  \nonumber \\
\end{align}



\begin{lemma}
$\kappa_1=\kappa_2$ holds for any mesoscopic mechanical model.
\end{lemma}
\begin{proof}
It is enough to show that 
\begin{align*}
\int_0^1 \int_0^1 & (\a-\frac{1}{2})(\a-\beta)  \tilde{W}(\a, \beta )  d\beta d\a   =\int_0^1 \int_0^1 \frac{(\a-\beta)^2}{2} \tilde{W}(\a, \beta ) d\beta d\a.
\end{align*}
By the assumption (iii), $\tilde{W}(\a, \beta )=\bar{W}(\a, \beta )(\a(1-\a))^{\frac{d}{2}-1} = \bar{W}(\b, \a )(\b(1-\b))^{\frac{d}{2}-1}=\tilde{W}(\b, \a)$. Therefore,
\begin{align*}
\int_0^1 \int_0^1 (\a-\beta)  \tilde{W}(\a, \beta )  d\beta d\a = \int_0^1 \int_0^1 (\a-\beta)  \tilde{W}(\beta, \a )  d\beta d\a
\end{align*}
 and then by the change of the variables, we get 
 \begin{equation}\label{eq:identity}
 \int_0^1 \int_0^1 (\a-\beta)  \tilde{W}(\a, \beta )  d\beta d\a=0.
 \end{equation}
  In the same manner, we can show that
 \begin{align*}
\int_0^1 \int_0^1 \a^2 \tilde{W}(\a, \beta )  d\beta d\a = \int_0^1 \int_0^1 \b^2 \tilde{W}(\a, \beta )  d\beta d\a
\end{align*}
and so we complete the proof.
\end{proof}

\begin{remark}
Using two different expressions (\ref{eq:static}) and (\ref{eq:static2}) of $\kappa_s(T)$, we can also show $\kappa_1=\kappa_2$.
\end{remark}

\begin{lemma}
$\kappa_f=\kappa_1$ if any only if 
\begin{align*}
\int_0^1 \int_0^1 &  \tilde{W}(\a, \beta ) d\beta d\a   =(d+\frac{3}{2})(d+\frac{1}{2})\int_0^1 \int_0^1 \a (\a-\beta)  \tilde{W}(\a, \beta ) d\beta d\a.
\end{align*}
\end{lemma}
\begin{proof}
It is clear by the explicit expressions (\ref{eq:constants}) of $\kappa_f$, $\kappa_1$ and (\ref{eq:identity}).
\end{proof}

\section*{Acknowledgement}

The author expresses her sincere thanks to Professor Herbert Spohn and Professor Stefano Olla for their insightful discussions and encouragement.

\end{document}